\documentclass[letterpaper, 10 pt, conference]{ieeeconf}  

\IEEEoverridecommandlockouts                              
\usepackage{graphics} 
\usepackage{epsfig} 
\usepackage{anyfontsize,epstopdf,placeins,subcaption}
\usepackage{amsmath,amssymb} 
\newtheorem{theorem}{Theorem}

\newtheorem{lemma}[theorem]{Lemma}

\def\one{\mbox{1\hspace{-3.85pt}\fontsize{11}{14.4}\selectfont\textrm{1}}}

\title{\LARGE \bf
Differential Games in Spread of Covid-19
}

\author{Sushant Vijayan
\thanks{This work was supported by Department of Atomic Energy, Goverentment of India.}
\thanks{Sushant Vijayan is with School of Technology and Computer Science,
        Tata Institute of Fundamental Research, Mumbai, India
        {\tt\small sushant.vijayanq@tifr.res.in}}%
}

\begin{document}

\maketitle
\thispagestyle{empty}
\pagestyle{empty}

\begin{abstract}

Given the ongoing Covid-19 pandemic, it is of interest to understand
how the infections spread as the combined result of measures taken by central planners (governments) and individual behavior. In this work, the spread of 
Covid-19 is modelled as a differentiable game between the planner and population with appropriate disease spread dynamical equations. 
We first characterise the equilibrium dynamics of only the population with modifed Susceptible-Infected-Recovered (SIR) equations to highlight the qualitative nature of the equilbrium.
Using this result, we formulate the joint equilibrium exposure profile between the planner and population.
Additionally, as in case of Covid-19, the role of asymptomatic carriers, inadequacies in testing, contact tracing and quarantining can lead to a significant underestimate of the true infected numbers as compared to just the detected numbers. Therefore, it is vital to model the true infected numbers within 
the context of choices made by individuals within the population. To incorporate this, we extend our framework by modifying the dynamics 
to include additional sub-compartments of `undetected infected' and `detected infected' in the disease dynamics. The individuals make their own estimates of the total infected 
from the detected numbers and base their strategies on those estimates. We show that these considerations lead to a retarded optimal 
control problem for the players. We present some simulation results based on these results to demonstrate how population behavior, planner control, detection rates and trust in the reported numbers play a key role in how the disease spreads.
\end{abstract}

\section{INTRODUCTION}

Infectious diseases spread because of interactions between the infected and the
susceptible. At an individual level, a simple strategy to reduce the possibility of
transmission is to voluntarily reduce ones interaction with others, i.e., to do social distancing. 
A central planner aims to impose constraints to individual behavior 
so as to maximise the total societal welfare. To model the disease spread effectively it is important 
to combine the choices of both the planner and the individuals in a unified way. The main goal of this work is to formulate a
game theoretic framework in which one can analyse and characterise
the resulting equilibrium between the individuals and the planner. A secondary goal 
is to modify the disease spread dynamics so as to account for the spread of disease
by infected individuals who are not detected and isolated from the susceptibles. This also leads to 
incorporating the individual's estimate of infection in the disease spread model. Using this, we present some simulations to qualitatively demonstrate the impact of planner control, population choices, detection rates and trust in the detected numbers in the spread of the disease.
\subsection{Related Work}
In mathematical epidemiology the spread of diseases is often modelled through
various compartmental models. The simplest of them is the SIR model \cite{c1}.
A considerable literature has been built up to include many extensions and variations to 
this basic model (see for example \cite{c2}, \cite{c3}, \cite{c4} and references therein).
An issue with these models is that they don't capture interventions of government nor individual 
choices. These decisions can have a significant impact on the 
disease spread trajectory.

Prior to the outbreak of Covid-19 some applications of optimal control in field of 
epidemiology include masking rates to prevent swine flu \cite{c5}, treatment rates in dengue transmission \cite{c6},
etc. Since the Covid-19 pandemic began there have been a considerable number of works
which formulate optimal Non-Pharmaceutical Intervention (NPI) as a control problem. \cite{c7} proposed a 
a lockdown that tapers down gradually.  \cite{c8} has a multi-group SIR model in
which the authors look at the optimal control problem for a social planner with control to do
age-specific targeted lockdowns. They show that the optimal solution is to enforce
stringent lockdowns for the older section of the population. \cite{c9} shows that
intermittent lockdowns may be better than moderate measures suggested above for a
class of utility functions particularly in low sero-prevalence scenarios. In a line of work closely related to the current work, \cite{c10} and \cite{c11} combine game theoretic equilibrium analysis based on
utility considerations of the individuals with the SIR model. It is clear from disease trajectories of many countries that it is not sufficient to
study a control problem for a planner or the equilibrium strategies of the population separately.
In all these models, they don't study combined interaction between individual choices and 
governmental policies. We find that the planner can try to take advantage of social distancing tendencies of individuals to control the spread of diseases.
The above models also don't incorporate how perceptions of the extent of disease spread affects the further spread of the disease.
\section{Preliminaries}
\subsection{Game theoretic setup}
We assume that the game is played till a finite horizon time $T$. This $T$ can be interpreted as the idealized vaccine arrival time wherein the vaccine affects 
the entire population instantly and puts an end to the disease.

The planner tries to control the spread of disease by imposing constraints on individual exposure choices. The individual tries to modify their behavior by either reducing or increasing their exposure fraction at any time $t$ while complying with the constraints imposed by the planner.

A strategy $A \in C[0,T]$ \footnote{$C[0,T]$ refers to the space of continuous functions taking values from $[0,T]$ to $\mathbb{R}$.} of the planner is such that $A_t$\footnote{$A_t=A(t)$, i.e, the value of the function $A$ at time $t$. Similarly for $g_t$. }$\in [0,1],$ $\forall t \in [0,T]$. $A_t$ sets the maximum permissible exposure of an individual at time $t$. $A_t$ represents the restrictions imposed by the planner on the individual's 
exposure profiles in the form of lockdowns, closing down schools, restricting public transport and other NPIs. $A_t$ can vary between $[0,1]$  but, as will be argued later, a natural upper bound of $A_t$ to be binding is the so called `population equilibrium' formed by individuals amongst themselves. In a realistic setting it is unacceptable (due to public disapproval of harsh lockdowns) for the planner to keep the threshold extremely low for extended periods of time. To capture this we extend the model to ensure certain average threshold limits are imposed on the planner.

An individual is assumed to respond to the spread of the disease by modifying their exposure to other individuals while complying with the threshold ($A_t$) imposed by the planner.
To keep the analysis simple, each individual is considered indistinguishable from another (one could consider the more realistic setting of several distinct groups) and is assumed to symmetrically employ an exposure strategy (or individual control or exposure profile) $g \in C[0,T]$ with $g_t \in [0,A_t]$, $\forall t \in [0,T]$. $g_t$ represents the reduced exposure from a normal baseline of unity prior to the onset of the disease.
For the purposes of calculating the equilbrium we shall, at times, also consider the strategy $g^{\alpha}$ of a canonical individual $\alpha$ differing from the symmetric strategy $g$ employed by the 
rest of the population.

In the simplest setting we assume that infected individuals are not isolated from the rest of the population and know the total 
number of infected at any given time. Later, we remove these restrictions by modelling the detecting of infection by introducing new infection compartments of `undetected infected' and `detected infected'.
\subsection{Dynamics of Disease Spread}
The evolution of susceptible fraction $S_t \in [0,1]$ and the infected fraction $I_t \in [0,1]$  is given by:
\begin{equation}
\label{eqn:dynbasic}
\begin{aligned}
\frac{d S_t}{dt}&=-\beta g_t^2 S_t I_t\\
\frac{dI_t}{dt}&=\beta g_t^2 S_t I_t-\gamma I_t
\end{aligned}
\end{equation}
 with the initial conditions $S_0=1-\epsilon$, $I_0=\epsilon$. $\epsilon$ is the initial fraction of infection in the population,
$\beta$ is the probability of getting infected per interaction with an infected individual and $\gamma$ is the recovery rate from infection\footnote{$\gamma, \beta, \epsilon, B,C$ and $R$ are positive real constants.}.
If the entire population plays a uniform exposure fraction $g_t$, the effective susceptible and infection fractions are $g_t S_t$ and $g_t I _t$ respectively,
 and the total number of interactions between the susceptibles and infected is $g_t^2S_tI_t$. This explains the quadratic dependence on $g_t$.
 We note that if all individuals played $g_t=1$, $\forall t \in [0,T]$ then we get the standard SIR model.

A canonical individual $\alpha$, will get infected at an unknown random time $\tau _{\alpha}$. We assume that each individual has an
estimate for their own survival probability which they estimate through a hazard rate model
\begin{equation}
\label{eqn:survprob}
\frac{dP(\tau_{\alpha}>t)}{dt}=-\beta g^{\alpha}_t g_t I_t P(\tau_{\alpha}>t)
\end{equation}
Here $g^{\alpha}_t$ is the exposure fraction strategy played by the individual $\alpha$ and $P(\tau_{\alpha}>t)$ is the survival probability\footnote{1-$P(\tau_{\alpha}>t)$ is the probability of getting infected before $t$.} at time $t$.
This estimate of the infection/survival probability and the trade off between benefits and risks of exposure will drive the individual's exposure strategy.
To model the benefits and risks for the players we next define the cost functionals which they each will seek to minimise.
\subsection{Cost Functionals of Players}
For simplicity, we assume $\alpha$ gets a linear rate of benefit $B$ per unit time from interacting with other individuals in the society. We also assume 
that upon contracting the disease, the individual suffers a one time cost $C$.
If $\alpha$ survives till $T$, there is a reward $R$ for surviving. Thus, $\alpha$ minimises the following cost functional $J_{\alpha}$ where ( in what follows $\one_{\{E\}}$ is the characteristic function of a set $E$.)
\begin{equation*}
 J_{\alpha}(g,g^{\alpha},A)=\mathbb{E}\bigg[-\int \displaylimits_{0}^{T \wedge \tau_{\alpha}}B g^{\alpha}_s ds +C\one_{\{\tau_{\alpha} \leq T\}}-R\one_{\{\tau_{\alpha} > T\}}\bigg].
\end{equation*}
The expectation above is with respect to the $\alpha$'s survival probability defined in (2) above. An equivalent formulation of the functional\footnote{This form shows the explicit dependence of $J^{\alpha}$ on $g,g^{\alpha}$ and $A$.} more 
suited for optimal control methods we seek to apply is
\begin{equation}
\label{eqn:indcost}
\begin{aligned}
 J_{\alpha}(g,g^{\alpha},A)=&\int \displaylimits_{0}^{T}P(\tau_{\alpha}>t)g^{\alpha}_t\bigg \{-B +C\beta g_t I_t\bigg\} dt\\
&-R P(\tau_{\alpha}>T).
\end{aligned}
\end{equation}
In the objective functionals above, we impose the restriction $g^{\alpha}_t,g_t \leq A_t$. 

An individual gets a benefit of $-Bg_t\Delta t$ and has an expected cost of infection $C\beta g_t^2I_t$ from an exposure strategy of $g_t$ in the interval $[t,t+\Delta t]$.
Thus, a measure of cost borne by the entire society is then $(-B +C\beta g_t I_t)g_tS_t\Delta t$ during this interval. Similarly, the total societal reward for surviving is $RS_T$. The above discussion motivates the planner's functional to be 
\begin{equation*}
J_{P}(g,A)=\int \displaylimits_{0}^{T}S_t g_t\bigg \{-B +C\beta g_t I_t\bigg\} dt-R S_T
\end{equation*}

\addtolength{\textheight}{-3cm}   

\section{MAIN RESULTS}
We shall first describe the result when there is no planner control ($A_t=1, \forall t$). The result highlights the qualitative nature of
the equilibrium amongst only the individuals in the population. It will also serve as a natural constraint on $A_t$ when we consider the 
more involved case with the planner control. This allows one to take a symmetric view on strategies - choices of the planner constrains the choices of the individual and 
vice versa.
\subsection{Pure Population Equilibrium Without Detection}
We shall assume that there is no control from the planner, i.e.,  $A_t=1$ for all $t \in [0,T]$.
It is only a game between the individuals of the population. An equilibrium result similar to that mentioned in this subsection can be found in \cite{c10}. To derive the equilibrium 
exposure strategy we use the Pontryagin Minimum Principle (PMP) (see section 3.3 in \cite{c12}). Let $\alpha$ play the strategy profile $g^{\alpha}_t$ and let the rest of the population 
play $g_{eq,t}$. Then in equilibrium we have
\begin{equation}
\label{eqn:nasheq}
\begin{aligned}
J_{\alpha}(g_{eq},g^{\alpha}_{eq},1)& \leq J_{\alpha}(g_{eq},g^{\alpha},1)\\
g^{\alpha}_{eq,t}&=g_{eq,t}
\end{aligned}
\end{equation}
The first condition is the definition of an equilibrium (Nash) while the second equation follows from the 
symmetric assumption on the individual's exposure profile. 

\begin{theorem}\label{thm1}
  For the dynamical game without a central planner the equilibrium exposure profile must be of the form:
  \begin{equation}
  \label{eqn:popeq}
      g^{\alpha}_{eq,t}=g_{eq,t}=\min\bigg(\frac{B }{\beta I_t (C-\lambda_{t})},1\bigg)\one_{\{C>\lambda_{t}\}}+\one_{\{C\leq \lambda_{t}\}}.
  \end{equation}
  The corresponding dynamics are governed by the equations:
  \begin{equation*}
      \begin{aligned}
      \frac{dS_t}{dt}&=-\beta g_{eq,t}^2 S_tI_t,\\
      \frac{dI_t}{dt}&=\beta g_{eq,t}^2 S_tI_t-\gamma I_t,\\
      \frac{d \lambda_{t}}{dt}&=g_{eq,t} (B-\beta I_t g_{eq,t}(C-\lambda_t)),
      \end{aligned}
  \end{equation*}
  with boundary conditions: $S_0=1-\epsilon,I_0=\epsilon,\lambda_{T}=-R$.
\end{theorem}
\begin{proof}
Let  $P_t:=P(\tau_{\alpha}>t).$ The Hamiltonian (dynamical system is (\ref{eqn:dynbasic})-(\ref{eqn:survprob}) and (\ref{eqn:indcost}) is the cost functional) for the individual $\alpha$'s minimisation problem when all the other individuals 
play the profile $g_{eq,t}$ is
\begin{equation*}
\begin{aligned}
&H(g^{\alpha}_tS_t,I_t,P_t,\lambda_t,\mu_t,\nu_t):= g^{\alpha}_t(-B +C\beta g_{eq,t} I_t)P_t+\\
&\lambda_{t}(-\beta g^{\alpha}_t g_{eq,t} I_t P_t)+\mu_{t}(-\beta g_{eq,t}^2 S_tI_t)+\nu_{t}(\beta g_{eq,t}^2 S_t I_t-\gamma I_t).
\end{aligned}
\end{equation*}
$\lambda_t, \mu_t ,\nu_t$ are the adjoint functions. We know that the optimal control $g^{\alpha}_t$ minimises the Hamiltonian and combining this along with (\ref{eqn:nasheq}) gives 
  \begin{equation*}
   \begin{aligned}
 & g^{\alpha}_{eq,t}=\underset{g^{\alpha}_t \in[0,1]}{\arg \min} \hspace{0.15cm} (-B+C \delta g_{eq,t} I_t-\lambda_{t}\beta  g_{eq,t} I_t )P_t g^{\alpha}_t , \\
\implies& g^{\alpha}_{eq,t}=\min\bigg(\frac{B}{\beta I_t (C-\lambda_{t})},1\bigg)\one_{\{C>\lambda_{t}\}}+\one_{\{C\leq \lambda_{t}\}}.
   \end{aligned}
  \end{equation*}
The only adjoint variable in the expression above is $\lambda_t$ and from PMP we have 
its dynamical equation to be
\begin{equation*}
\begin{aligned}
\frac{d \lambda_t}{dt}=&-\frac{\partial H}{\partial P},\\
=&g_{eq,t} (B-\beta I_t g_{eq,t}(C-\lambda_t)).
\end{aligned}
\end{equation*}
with the additional boundary condition for this equation supplied by the transversality condition of PMP, i.e., $\lambda_T=-R$.
\end{proof}
From standard results in optimal control (see \cite{c15}) we can associate the adjoint variable with partial derivative of the value function wrt state variables. Here, 
$\lambda_t$ is the partial derivative of the value function wrt $P_t$. This leads us to conclude (see Appendix A for proof) that
\begin{equation}
\label{adjinterp}
\lambda_t=\mathbb{E}\bigg[\int \displaylimits_{t}^{T \wedge \tau_{\alpha}}-B g^{\alpha}_s ds +C\one_{\{\tau_{\alpha} \leq T\}}-R\one_{\{\tau_{\alpha} > T\}}\bigg|\tau_{\alpha}>t\bigg]
\end{equation}
Thus $\lambda_t$ can be interpreted as the future expected cost given the individual has survived till time t. With this interpretation the 
equilibrium profile in Theorem \ref{thm1} implies that whenever either the total number of infections  or the future expected costs are high then 
the population starts social distancing until the number of infections decrease.

\begin{lemma}\label{lem1}
$\lambda_t$ is a non-decreasing function of $t$. More precisely whenever $g_{eq,t}<1$ then $\lambda_t$ is a constant and 
is strictly increasing whenever $g_{eq,t}=1$.
\end{lemma}
\begin{proof}
This is easily verified by looking at the expressions for $g_{eq,t}$ and $\frac{d \lambda_t}{dt}$ from Theorem \ref{thm1}.
\end{proof}
This monotonicity intuitively is because for any arbitrarily small interval the equilibrium strategy must accrue more benefit than a strategy of complete social distancing, i.e, $g_t=0$ within the same time interval.

Lemma \ref{lem1} indicates that in equilibrium the dynamics is a hybrid one with the value of $g_{eq,t}$ triggering the switch 
between the states of social distancing and normal behavior. When $g_{eq,t}=1$ we have normal behavior with no social distancing and the dynamics of $S_t,I_t$ is just the same as the SIR model.
We characterise the dynamics in social distancing regime $(g_{eq,t}<1)$  with the following lemma:
\begin{lemma}\label{lem2}
In the social distancing regime we have the following relation between $S_t$ and $I_t$
\begin{equation*}
    I_t=\sqrt{\frac{S_t}{C_0}}\bigg(\frac{K_1(2 \sqrt{C_0 S_t})-C_1L_1(2 \sqrt{C_0 S_t})}{K_0(2 \sqrt{C_0 S_t})+C_1L_0(2 \sqrt{C_0 S_t})}\bigg)
\end{equation*}
where $L_n(x)$ and $K_n(x)$ denote the $n^{th}$ order modified Bessel functions of the first and the second kind, respectively. $C_0=\frac{\gamma \beta (C-\lambda_0)^2}{B^2}$. $C_1$ is determined by the initial values $S_0,I_0,\lambda_0$ at the onset of social distancing.
\end{lemma}
\begin{proof}
In the social distancing regime from Theorem \ref{thm1} we have $g_{eq,t}=\frac{B }{\beta I_t (C-\lambda_{t})}$ and that $\lambda_t$ is 
constant. Thus we have:
\begin{equation*}
  \begin{aligned}
  \frac{dI}{dS}&=-1+\frac{\gamma}{\beta g_{eq,t}^2 S}\\
  &=-1+\frac{C_0 I^2}{S} \hspace{1cm}
  \end{aligned}
\end{equation*}
This is a Riccatti equation and the standard reduction of a Riccatti equation to a linear second order ODE (see \cite{c16}) gives the result.
\end{proof}
We shall use the results in this subsection to characterise the more complicated equilibrium that exists between the planner and population
of individuals.
\subsection{Equilibrium with Population and Planner}\label{B}
In this case the planner is trying to optimally set a threshold to minmise its own cost functional. The constraint on the planner is that if it sets 
a high enough threshold then individuals behavior may follow the pure population equilibrium from the previous subsection (and hence the threshold $A_t$ becomes non-binding). We have the constraints that 
$0 \leq g_t \leq A_t$ and $0 \leq A_t \leq g^{pop}_{eq,t}$ \footnote{$g^{pop}_{eq,t}$ refers to the pure population equilibrium described in the (\ref{eqn:popeq}).}.
In this case at equilibrium we must have 
\begin{equation}
\label{plannerneq}
\begin{aligned}
& J_{\alpha}(g_{eq},g^{\alpha}_{eq,t},A_{eq,t}) \leq J_{\alpha}(g_{eq},g^{\alpha},A_{eq,t}),\\
&J_{P}(g_{eq},A_{eq}) \leq J_{P}(g_{eq},A),\\
&g^{\alpha}_{eq,t}=g_{eq,t}. 
\end{aligned}
\end{equation}
 As the set of admissible controls for the players vary with both time
and state we use a generalised version of PMP (Theorem 3.1 in \cite{c13}). We have the following result characterising the 
equilibrium between population and planner:
 \begin{theorem}\label{thm2}
 The population-central planner game has the following equilibrium profile:
\begin{equation}
\label{planpopeq}
  \begin{aligned}
  g_{eq,t}=& \min\bigg(\frac{B}{\beta I_t (C -\lambda_{t})},A_{eq,t}\bigg)\one_{\{C>\lambda_{t}\}}+A_{eq,t}\one_{\{C\leq \lambda_{t}\}}\\
 A_{eq,t}=& \min \bigg(\frac{B}{2 \beta I_t (	C-\lambda_{1,t}+\lambda_{2,t})},g^{pop}_{eq,t}\bigg)\one_{\{C +\lambda_{2,t}>\lambda_{1,t}\}}\\
&+g^{pop}_{eq,t}\one_{\{C +\lambda_{2,t} \leq \lambda_{1,t}\}} 
  \end{aligned}
 \end{equation}
 with $g^{pop}_{eq,t}$ as defined in (\ref{eqn:popeq}) of Theorem \ref{thm1}. The corresponding dynamics is given by:
 \begin{equation*}
 \begin{aligned}
     \frac{d S_t}{dt}&=-\beta g_{eq,t}^2 S_t I_t,\\
     \frac{d I_t}{dt}&=\beta g_{eq,t}^2S_t I_t-\gamma I_t,\\
     \frac{d\lambda_t}{dt}&=g_{eq,t}(B -\beta g_{eq,t} I_t(C-\lambda_t )),\\
      \frac{d \lambda_{1,t}}{dt}&=g_{eq,t}(B -\beta g_{eq,t} I_t(C +\lambda_{2,t}- \lambda_{1,t})),\\
      \frac{d \lambda_{2,t}}{dt}&=-\beta g_{eq,t}^2 S_t (C  -\lambda_{1,t}+\lambda_{2,t})+\gamma \lambda_{2,t}.    
  \end{aligned}
 \end{equation*}
 with boundary conditions:  $S_0=1-\epsilon,I_0=\epsilon,\lambda_T=-R,\lambda_{1,T}=0,\lambda_{2,T}=0.$
\end{theorem}
\begin{proof}
The individual is trying to minimise $J_{\alpha}$ given the strategies $g_{eq,t}$ and $A_{eq,t}$. Assuming $A_{eq,t}$ to be a given function of time we can parametrize the admissible set of controls as $Q_1 \leq 0$ where:
\begin{equation*}
    Q_1(t,S_t,I_t,g^{\alpha}_t)=g^{\alpha}_t-A_{eq,t}.
\end{equation*}
This has a non-zero derivative wrt the control and hence we can apply Theorem 3.1 from \cite{c13} for the individual's control problem. The difference from Theorem \ref{thm1} is that the controls are restricted dynamically. The Hamiltonian has to be minimised only within this dynamically changing feasible control set. The Hamiltonian in this case is
\begin{equation*}
    \begin{aligned}
     H(S_t,I_t,&P_t,g^{\alpha}_t,\lambda_t,\lambda_2,\lambda_3,\mu_t):=P_t g^{\alpha}_t(-B +C\beta g_{eq,t} I_t)\\
     &+\lambda_{t}(-\beta g^{\alpha}_t g_{eq,t} I_t P_t)+\kappa_{t}(-\beta g_{eq,t}^2 S_tI_t)\\
     &+\iota_{t}(\beta g_{eq,t}^2 S_tI_t-\gamma I_t)+\mu_t(g^{\alpha}_t-A_{eq,t}).
    \end{aligned}
\end{equation*}
$\lambda_t,\kappa_t,\iota_t$ and $\mu_t$ are adjoint variables. Minimising the Hamiltonian with (\ref{plannerneq}) gives
  \begin{equation*}
   \begin{aligned}
 g^{\alpha}_{eq,t}=\underset{g^{\alpha}_t \in[0,A_{eq,t}]}{\arg \min}& \hspace{0.25cm}P_t g^{\alpha}_t (-B+C \beta g_{eq,t} I_t-\lambda_{t}\beta f_t g_{eq,t} I_t) \\
 g^{\alpha}_{eq,t}=g_{eq,t}&=\min\bigg(\frac{B}{\beta I_t (C -\lambda_{t})},A_{eq,t}\bigg)\one_{\{C>\lambda_{t}\}}\\
&+A_{eq,t}\one_{\{C\leq \lambda_{t}\}}.
   \end{aligned}
  \end{equation*}
In the first step above $\mu_t$ doesn't appear because of complementarity condition $\mu_t(g^{\alpha}_t-A_{eq,t})=0$. The exposure profile in this case is the population equilibrium with upper threshold now set to $A_{eq,t}$ rather than 1.
 
In case of the  planner, for the threshold to be binding, it must be set lesser than the population equilibrium profile (see (\ref{eqn:popeq})). 
\begin{equation*}
g^{pop}_{eq,t}=\min\bigg(\frac{B}{\beta I_t (C -\lambda_{t})},1\bigg)\one_{\{C>\lambda_{t}\}}+\one_{\{C\leq \lambda_{t}\}}
\end{equation*}
Hence, for the planner, the state variables are $S_t,I_t,\lambda_t$. The  planner minimises $J_P$ given the population strategy $g_{eq,t}$.
The dynamical equations relevant to the planner are:
  \begin{equation*}
      \begin{aligned}
      \frac{dS_t}{dt}&=-\beta A_t^2 I_tS_t,\\
      \frac{dI_t}{dt}&=\beta A_t^2 I_tS_t-\gamma I_t,\\
      \frac{d \lambda_{t}}{dt}&=A_t (B-\beta I_t A_t(C-\lambda_t)).
      \end{aligned}
  \end{equation*}
The set of admissible controls for the planner can be summarised by  $Q_2 \leq 0$ where
\begin{equation*}
Q_2(t,S_t,I_t,\lambda_t,A_t)=A_t-g^{pop}_{eq,t}.
\end{equation*}
Thus we can again invoke Theorem 3.1 from \cite{c13} for the planners control problem, The Hamiltonian is given by
\begin{equation*}
    \begin{aligned}
     H(S_t,I_t,\lambda_t&,A_t,\lambda_1,\lambda_2,\lambda_3,\mu_t):=S_t A_t(-B +C\beta A_t I_t)\\
     &+\lambda_{1,t}(-\beta A_t^2 I_t S_t)+\lambda_{2,t}(\beta A_t^2 S_tI_t-\gamma I_t)\\
     &+\lambda_{3,t}(B A_t-\beta A_t^2I_t(C-\lambda_t))+\mu_t(A_t-g^{pop}_{eq,t}).
    \end{aligned}
\end{equation*}
Minimising the Hamiltonian along with (\ref{plannerneq}) gives
  \begin{equation*}
   \begin{aligned}
  A_{eq,t}=\underset{A_t \in[0,g^{pop}_{eq,t}]}{\arg \min} & B A_t (\lambda_{3,t}-S_t)+\beta I_t S_t A_t^2(C-\lambda_{1,t} \\
&+\lambda_{2,t}-\lambda_{3,t}(C -\lambda_t))
   \end{aligned}
  \end{equation*}
The adjoint variables for the planner are denoted by $\lambda_{1,t},\lambda_{2,t},\lambda_{3,t}$.
The adjoint equation of $\lambda_{3,t}$ becomes:
 \begin{equation*}
     \begin{aligned}
      \Dot{\lambda_{3,t}}=-\beta A_t^2S_t I_t
     \end{aligned}
 \end{equation*}
 with the boundary condition $\lambda_{3,T}=0$. But as $S_t, I_t$ are positive, the only way this boundary condition can be satisfied is when $\lambda_{3,t}=0,\forall t$. Using this in the minimum principle we get:
 \begin{equation*}
     \begin{aligned}
      A_{eq,t}&=\underset{A_t \in[0,g^{pop}_{eq,t}]}{\arg \min}-B A_t+\beta I_t S_t A_t^2(C -\lambda_{1,t}+\lambda_{2,t})\\
      A_{eq,t}&=\min \bigg(\frac{B }{2 \beta I_t (C-\lambda_{1,t}+\lambda_{2,t})},g^{pop}_{eq,t}\bigg)\one_{\{C +\lambda_{2,t}>\lambda_{1,t}\}}\\
	&+g^{pop}_{eq,t}\one_{\{C +\lambda_{2,t}\leq \lambda_{1,t}\}}. 
     \end{aligned}
 \end{equation*}
 The adjoint equations become:
 \begin{equation*}
     \begin{aligned}
      \Dot{\lambda_{1,t}}&=B A_{eq,t} -\beta A_{eq,t}^2 I_t (C -\lambda_{1,t}+\lambda_{2,t}),\\
      \Dot{\lambda_{2,t}}&=-\beta A_{eq,t}^2 S_t (C -\lambda_{1,t}+\lambda_{2,t})+\gamma \lambda_{2,t}.
     \end{aligned}
 \end{equation*}
\end{proof}
It can be easily seen that $A_{eq,t}=g_{eq,t}$ and hence the planner's threshold is always binding. The exposure profile of the population is also seen to be the net result of the strategic choices of the planner and the 
population.\\
Additionally, we impose an average threshold constraint on the planner ie. $\int_{0}^{T}A_t dt >C_1$. This is to prevent the planner from accessing strategies which entail harsh thresholds over an extended period. These types of constraints are called "isoperimetric  constraints" and are handled in a standard way in optimal control literature (see \cite{c14}).
\subsection{Detection of Infection}
As mentioned earlier it is important to model the group of undetected infectious indivduals who spread the disease. This framework also allows us to model the estimates of infection spread made by individuals and the planner.
In this section we partition the infected group $I_t$ into two subgroups - $I_{u,t}$, the undetected group of infected and $I_{d,t}$, the detected group of infected. We have:
\begin{equation}
I_t=I_{u,t}+I_{d,t}.
\end{equation}
 We assume that once the infected are detected they are effectively quarantined and no longer infect the susceptibles. Hence, an infected individual either recovers without being detected or gets quarantined after detection. An infected individual is modelled to remain infectious for a period of $\frac{1}{\gamma}$ and has a probability of being detected in this period. For an individual $\alpha$, conditioned on the event $\tau_{\alpha}=t$, we assume a probability density of detection over the period $(t,t+ \frac{1}{\gamma}]$. $\tau_{d}$ denotes  the random time of detection once $\alpha$ is infected. Thus $\tau_d \in [0,\frac{1}{\gamma}]$. For simplicity, we shall assume that the probability of detection is uniform over $[0,\frac{1}{\gamma}]$.Thus, the individual's objective functional becomes:
\begin{equation*}
\mathbb{E}[\int \displaylimits_{0}^{(\tau_{\alpha}+\tau_d )\wedge T}-B g^{\alpha}_t dt+C \one_{\{\tau_{\alpha} \leq T\}}-R \one_{\{\tau_{\alpha} >T\}}]
\end{equation*}
We can re-write the first term as:
\begin{equation*}
-\int_{0}^{T}B g^{\alpha}_t  P\bigg((\tau_{\alpha}+\tau_d)\wedge T>t\bigg)dt.
\end{equation*}
 We have $\forall t \in [0,T]$:
 \begin{equation*}
 P\bigg((\tau_{\alpha}+\tau_d)\wedge T>t\bigg)=1-P\bigg(\tau_{\alpha}+\tau_d \leq t\bigg),
 \end{equation*}
with
\begin{equation}
  \begin{aligned}
    P\bigg(\tau_{\alpha}+\tau_d \leq t\bigg)&=P\bigg(\tau_{\alpha} \leq t-\frac{1}{\gamma}\bigg)\\
&+P\bigg(\tau_d \leq t-\tau_{\alpha} \leq  \frac{1}{\gamma}\bigg).
    \end{aligned}
\end{equation}
 Assuming that $\tau_{\alpha}$ has a density $f$ and the uniform conditional density for $\tau_d$ is $\frac{\gamma}{\eta}$, we rewrite the second term in RHS of (9) as:
 \begin{equation*}
    \begin{aligned}
 P\bigg(\tau_d \leq t-&\tau_{\alpha} \leq \frac{1}{\gamma}\bigg)=\int_{t-\frac{1}{\gamma}}^{t} \bigg(\int_{t-\frac{1}{\gamma}}^{r}f(s)ds\bigg) \frac{\gamma}{\eta}dr\\
 &=P\bigg(\tau_{\alpha}>t-\frac{1}{\gamma}\bigg) \frac{1}{\eta}-\frac{\gamma}{\eta}\int_{t-\frac{1}{\gamma}}^{t}P(\tau_{\alpha}>r)dr
    \end{aligned}
 \end{equation*}
 Here $\frac{1}{\eta}$ (with $\eta>1$) captures the probability of detection and is a parameter in the model. Setting $M_t:=\gamma \int_{t-\frac{1}{\gamma }}^{t}P(\tau_{\alpha}>r) dr,$  we rewrite the individual's objective functional as (superscript $d$ stands for detected):
 \begin{equation}
\begin{aligned}
  J^d_{\alpha}(g,g^{\alpha},A)&=-B\int_{0}^{T}g^{\alpha}_t \bigg\{P\bigg(\tau_{\alpha}>t-\frac{1}{\gamma}\bigg)\bigg(1-\frac{1}{\eta}\bigg)\\
&+\frac{M_t}{\eta} \bigg\}dt+C P(\tau_{\alpha}\leq T) -RP(\tau_{\alpha}>T). 
\end{aligned}
 \end{equation}
 The individual $\alpha$ has knowledge only of $I_{d,t}$ and makes an estimate of $I_{u,t}$ from $I_{d,t}$. For simplicity, we assume that the estimate has the form:
 \begin{equation*}
 \widehat{I_{u,t}}=\kappa I_{d,t}.
 \end{equation*}
$\kappa$ encapsulates the trust the population has on the reported detected numbers.  Now as $P(\tau_{\alpha}>t) $ is  linked to the individuals perception of infection, we must modify (2) to:
 \begin{equation*}
   \begin{aligned}
     \frac{dP(\tau_{\alpha}>t)}{dt}&=-\beta g^{\alpha}_t g_t \widehat{I_{u,t}} P(\tau_{\alpha}>t),\\
     &=-\beta g^{\alpha}_t g_t \kappa I_{d,t}P(\tau_{\alpha}>t).
     \end{aligned}
 \end{equation*}
 The state equations for the individual are:
\begin{equation}
\label{detdynamics}
\begin{aligned}
    \frac{dS_t}{dt}&=-\beta g_t^2 S_tI_{u,t},\\
    \frac{dI_{u,t}}{dt}&=\beta g_t^2 S_tI_{u,t}-\gamma (1+\frac{1}{\eta})I_{u,t},\\
    \frac{d I_{d,t}}{dt}&=\frac{\gamma}{\eta} I_{u,t}-\gamma I_{d,t},\\
    \frac{dP(\tau_{\alpha}>t)}{dt}&=-\beta g^{\alpha}_t g_t \kappa I_{d,t}P(\tau_{\alpha}>t),\\
    \frac{dM_t}{dt}&=\gamma \bigg(P(\tau_{\alpha}>t)-P\bigg(\tau_{\alpha}>t-\frac{1}{\gamma}\bigg)\bigg).
\end{aligned}
\end{equation}
New infections are caused by the interaction between the susceptibles and undetected infected. These new infections are intially always assumed to be undetected. Then, some of the undetected infected  move to $I_{d,t}$ due to the detection density $\frac{\gamma}{\eta}$. 

The control formulation now has constant delays in state variable $P_t (:=P(\tau_{\alpha}>t))$ for both the objective functional and state equations. These types of control problems are called Retarded Optimal Control Problems (ROCP). We shall use a version of the minimum principle for this ROCP (see theorem 4.2 in \cite{c13}). Although one can in principle also include the planner's control in this more elaborate model, the resulting profile is rather messy and unwieldy. This joint equilibrium profile can be derived in an analogous manner as in section \ref{B} and is omitted. We shall assume a control on the part of the planner and present the result for only the resulting population equilibrium under this control.

\begin{theorem}\label{thm3}
 The population game with detection has the following equilibrium profile:
\begin{equation}
\label{detpopeq}
  \begin{aligned}
 g_{eq,t}&=\min \bigg(\frac{B ((1-1/\eta)P_{t-1/\gamma}+M_t/\eta)}{\beta P_t \kappa I_{d,t}(C-\lambda_{t})},1\bigg)\one_{\{C>\lambda_{t}\}}\\
&+ \one_{\{C\leq \lambda_{t}\}},
  \end{aligned}
 \end{equation}

The equilibrium dynamics is given by (\ref{detdynamics}). Additionally, the equation for the adjoint variable $\lambda_t$ is given by
 \begin{equation*}
 \begin{aligned}
    \frac{d\lambda_{t}}{dt}&=- (\beta g_{eq,t}^2 \kappa I_{d,t}(C-\lambda_{t}) +\frac{\gamma}{\eta}(T-t))\\
     &+\one_{[0,T-1/\gamma]}(t)\bigg(B g_{eq,t+\frac{1}{\gamma}}\bigg(1-\frac{1}{\eta}\bigg)+\frac{\gamma}{\eta}\bigg(T-t-\frac{1}{\gamma}\bigg)\bigg),
  \end{aligned}
 \end{equation*}
 with  boundary conditions: $S_0=1-\epsilon,I_{u,0}=\epsilon,I_{d,0}=0,\lambda_T=-R.$ 
\end{theorem}
\begin{proof}
We apply theorem 4.2 from \cite{c13} to the ROCP of the individual $\alpha$. Compared to Theorems \ref{thm1} \& \ref{thm2} the Hamiltonian also incorporates the delayed state variables. Consequently, the adjoint equations and the optimal control depend on these delayed variables.The Hamiltonian is given by:
    \begin{equation*}
     \begin{aligned}
        H(S_t,I_{d,t},&I_{u,t},P_t,M_t,P_{t-1/\gamma},\lambda_{t}):=C\beta g^{\alpha}_t g_t \kappa I_{d,t}P_t\\
        &-B g^{\alpha}_t \bigg(P_{t-1/\gamma}\bigg(1-\frac{1}{\eta}\bigg)+\frac{1}{\eta}M_t\bigg)\\
        &+\lambda_{t}(-\beta g^{\alpha}_t g_t \kappa I_{d,t}P_t)+\lambda_{1,t}(-\beta g_t^2 S_t I_{u,t})\\
        &+\lambda_{2,t}\bigg(\beta g_t^2 S_t I_{u,t}-\gamma\bigg(1+\frac{1}{\eta}\bigg)I_{u,t}\bigg)\\
        &+\lambda_{3,t}\bigg(\frac{\gamma}{\eta} I_{u,t}-\gamma I_{d,t}\bigg)+\lambda_{4,t}(P_t-P_{t-1/\gamma})
        \end{aligned}
    \end{equation*}
Minimising the Hamiltonian as a function of $g^{\alpha}$ and using $g_{eq,t}=g^{\alpha}_{eq,t},$ we get 
  \begin{equation*}
   \begin{aligned}
  g^{\alpha}_{eq,t}=\underset{g^{\alpha}_t \in[0,1]}{\arg \min} \hspace{0.05cm}& -B g^{\alpha}_t \bigg(\frac{P_{t-1/\gamma}(\eta-1)+M_t}{\eta}\bigg) \\
&+(C-\lambda_t)g^{\alpha}_t g_t \kappa I_{d,t}P_t, \\
   \end{aligned}
  \end{equation*}
\begin{equation*}
\begin{aligned}
g^{\alpha}_{eq,t}=& \min \bigg(\frac{B ((1-1/\eta)P_{t-1/\gamma}+M_t/\eta)}{\beta P_t \kappa I_{d,t}(C-\lambda_{t})},1\bigg)\one_{\{C>\lambda_{t}\}}\\
&+ \one_{\{C\leq \lambda_{t}\}}
\end{aligned}
\end{equation*}
which is the profile in (\ref{detpopeq}). As the exposure profile depends only on $\lambda_{t}$, which in turn depends on $\lambda_{4,t}$, we consider differential equations of only these two variables. We can explicitly solve for $\lambda_{4,t}$ with the condition $\lambda_{4,T}=0$. We have 
\begin{equation*}
\begin{aligned}
& \frac{d \lambda_{4,t}}{dt}=-\frac{\partial H}{\partial M_t}     \\
& \frac{d \lambda_{4,t}}{dt}=-\frac{\gamma}{\eta} \\
& \lambda_{4,t}=\frac{\gamma}{ \eta}(T-t).
\end{aligned}
\end{equation*}
The differential equation for $\lambda_t$  is then given by:
\begin{equation*}
\begin{aligned}
 \frac{d \lambda_{t}}{dt}=&-\frac{\partial H}{\partial P_t}-\one_{[0,T-1/\gamma]}(t)   \frac{\partial H}{\partial P_{t-1/\gamma}}  \\
=&- (\beta g_{eq,t}^2 \kappa I_{d,t}(C-\lambda_{t}) +\frac{\gamma}{\eta}(T-t))+\one_{[0,T-1/\gamma]}\\
     &\bigg(B g_{eq,t+\frac{1}{\gamma}}\bigg(1-\frac{1}{\eta}\bigg)+\frac{\gamma}{\eta}\bigg(T-t-\frac{1}{\gamma}\bigg)\bigg).
\end{aligned}
\end{equation*}
\end{proof}
\section{SIMULATION RESULTS}
The theorems derived in section III provide a basis for simulating a dynamical system with initial infection. The equilibrium solution to the game leads to solving a system a differential equations with a two point boundary condition. 
This is fairly typical in optimal control and is due to the PMP. The equilibrium in the model with detection leads to a two point boundary value problem in a system of advanced-delay differential equations. We only approximately solve this system 
by using a cubic extrapolation for the advanced term (see \cite{c13,c17} for other numerical examples).

The boundary value problem was solved using a shooting approach coupled with an initial value differential and delay-differential equation solver. This then reduces the problem to solving a nonlinear problem (see \cite{c18}) of finding the appropriate initial values for the adjoint variables. 

\begin{table}[!ht]
\centering 
\begin{tabular}{c c } 
\hline 
Parameter & Value \\  
\hline 
T & 400\\
$\epsilon$ & 1.65e-08 \\ 
B & 0.01\\
C & 1 \\
R & 0\\
$\beta$ & 0.2\\ 
$\gamma$ & 0.05\\
\hline 
\end{tabular}
\label{table:params} 
\caption{Parameter values used in simulations.}
\end{table}
The values for the various parameters in the simulation are given in Table I above. The parameters $\eta$, the probability of detection, and $\kappa$, trust in detected numbers, are varied to give various scenarios shown in figures \ref{fig:detsus},\ref{fig:detexp} and \ref{fig:detinf}. 

\begin{figure}[h!]
  \includegraphics[width=\linewidth]{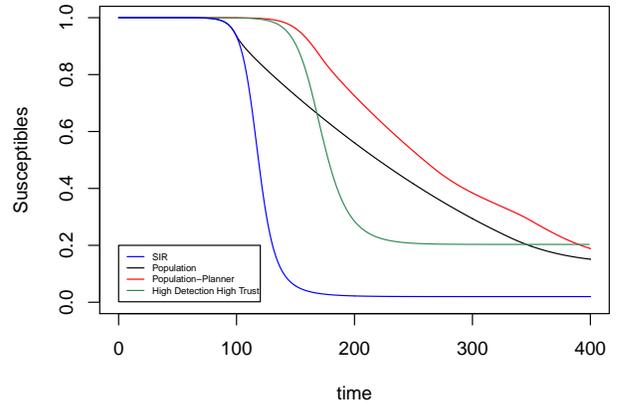}
  \caption{Susceptible fraction for SIR (blue), pure population equilibrium (black), planner-population equilibrium (red) and population equilbirium with high detection \& trust (green). }
  \label{fig:mainsus}
\end{figure}

\begin{figure}[h!]
  \includegraphics[width=\linewidth]{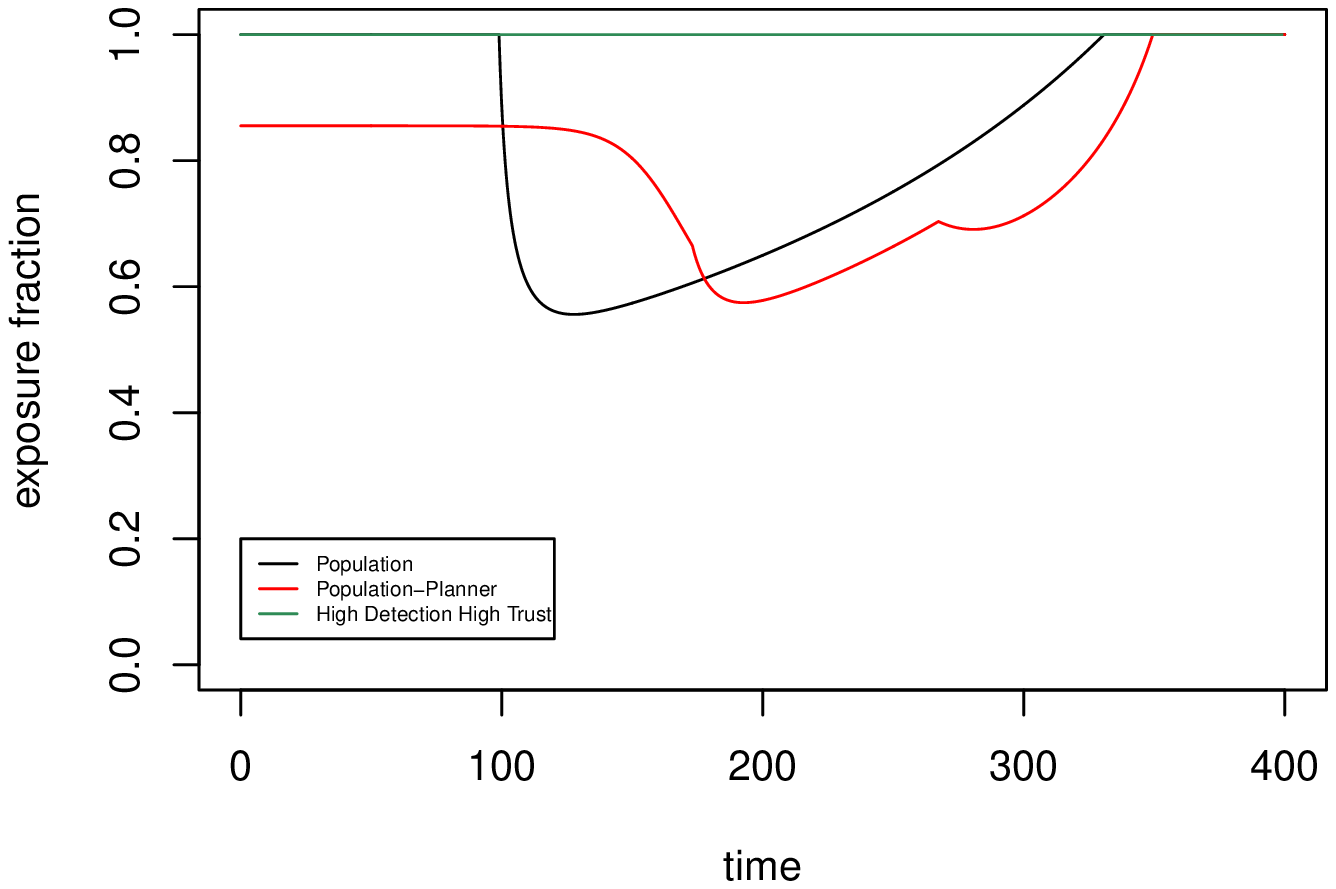}
  \caption{Exposure strategies for cases shown in figure 1.}
  \label{fig:mainexp}
\end{figure}

\begin{figure}[h!]
  \includegraphics[width=\linewidth]{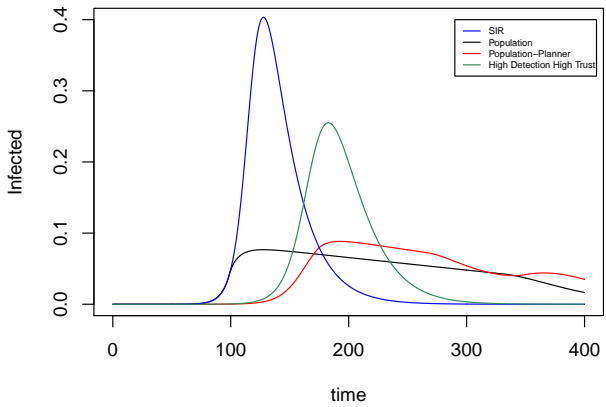}
  \caption{Infection fraction for cases shown in figure 1.}
  \label{fig:maininf}
\end{figure}

In figures \ref{fig:mainsus}, \ref{fig:mainexp} and \ref{fig:maininf}, we plot the susceptible fraction, exposure and infected fraction, respectively, versus time. SIR (blue) shows an exponential decrease in susceptibles at peak infection with almost no susceptibles remaining at the end. This is the worst case scenario- no social distancing, no detection, no quarantining and no planner control. It has the highest peak infection whose onset is advanced compared to other scenarios. 

The case with population equilibrium (black) shows a prolonged infection peak with much slower decrease in susceptible numbers (vis-a-vis SIR). From figure \ref{fig:mainexp} it is clear the population starts to socially distance as the infection numbers increase and only stops doing so when it is close to the vaccine arrival time (T).

In the case with planner control (red) the planner initially sets a moderate threshold (see figure \ref{fig:mainexp}) to control the spread of the disease. This results in a delayed infection peak. As the infection numbers inevitably rise the population voluntarily reduce exposure below even the planner's threshold. This results in the peak infection becoming plateaued in a manner similar to population equilibrium. The social distancing and thresholding is gradually reduced as we approach the vaccine arrival (which is assumed to instantly stop the infection). Compared to the population equilibrium case the peak infection is delayed and the economic impact (as measured by exposure time) is reduced.

The case with high detection ($\eta=1$) and high trust ($\kappa=1$) leads to significantly lower peak than SIR but unlike the population or population planner cases the peak is not prolonged (though the peak infection itself is higher). There is no social distancing due to high levels of trust, detection and quarantining. The total susceptible surviving at the end is similar to population or population-planner case. This seems to be the most preferable case where the peak infection is delayed and not prolonged and the economic impact minimal assuming the higher peak infection can be managed.

\begin{figure}[h!]
  \includegraphics[width=\linewidth]{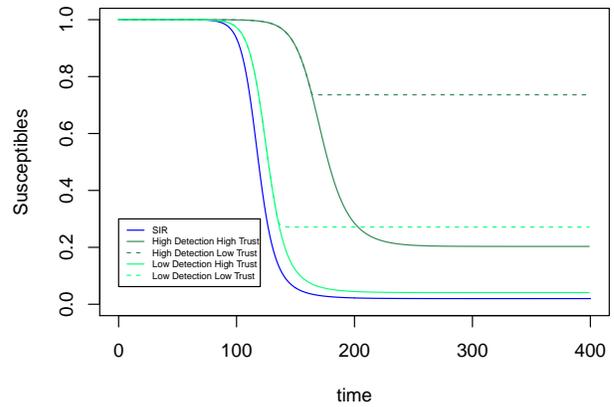}
  \caption{Susceptible fraction for SIR (blue), population equilibrium with high detection \& trust (dark green), with high detection \& low trust (dashed dark-green), with low detection \& high trust (light green) and low detection \& low trust (dashed light green).  }
  \label{fig:detsus}
\end{figure}

\begin{figure}[h!]
  \includegraphics[width=\linewidth]{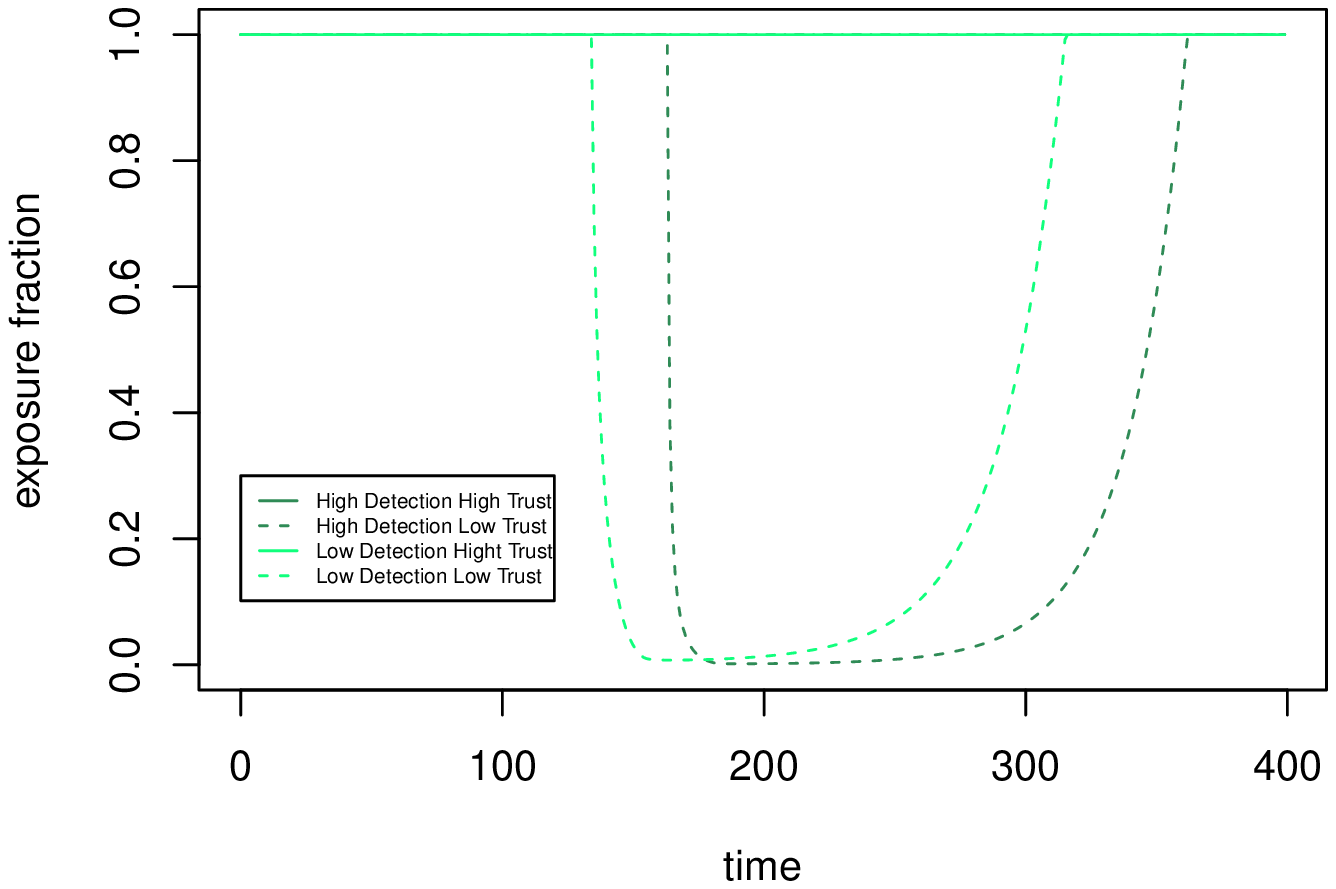}
  \caption{Exposure strategies  for cases shown in figure 4. }
  \label{fig:detexp}
\end{figure}

\begin{figure}[h!]
  \includegraphics[width=\linewidth]{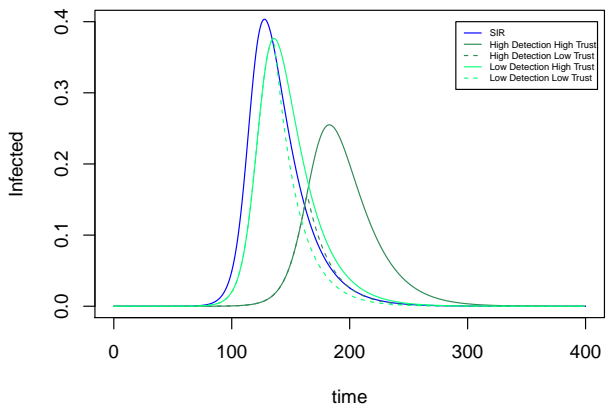}
  \caption{Infection fraction for cases shown in figure 4.}
  \label{fig:detinf}
\end{figure}

In figures \ref{fig:detsus}, \ref{fig:detexp} and \ref{fig:detinf}, we plot the effects of different detection rates ($\eta$) and trust parameters ($\kappa$) on the spread of disease. We have already discussed the case with high detection and high trust in the paragraph above. In the case where the detection rates are high ($\eta=1$) but trust ($\kappa=32$) is low (dashed dark green),  then we observe that (see figure \ref{fig:detexp}) as soon as infected numbers peak the population completely reduces exposure to zero. This completely stops the disease spread. This is due to the low trust in the detected numbers. The population believes the planner is doing a poor job of the detection even though in reality the detection rates are high. This leads to unnecessary loss of exposure benefits. 

In the low detection ($\eta =5$) but high trust ($\kappa=1$) (light green) case, the disease spread curve is very close to the SIR situation. This is expected since in this case there is poor detection and yet the population trusts the detected numbers are an accurate measure of disease spread. This leads to the undetected infected comprising the entirety of the infected numbers while the population seeing the low detection numbers chooses not to socially distance. This is similar to the SIR situation where there is no planner control nor any social distancing.

Finally in the low detection ($\eta=5$)  and low trust ($\kappa=32$) (dashed light green) case, just as in the high detection low trust case, the population reduces exposure completely as soon as the infection numbers start to peak. However, compared to the high detection setting the total infected numbers is higher because most of the infected aren't detected and help spread the disease. 

The various comparisons seem to suggest ideally for controlling disease spread one needs to have high detection rates with transparency in the reported numbers so that the population has confidence in the reported numbers. If the detection numbers are low and the planner tries to underplay the poor detection it could result in the worst possible scenario- an unmitigated disease spread which could stress the medical resources at peak infection. In practice, it is often difficult to achieve very high detection rates due to asymptomatic carriers, testing errors etc. and hence must be combined with some moderate amount of planner control in form of lockdowns and imposing social restrictions. The population also has an important role to play by voluntarily reducing exposure and other NPIs (wearing masks, sanitisation, adhering to restrictions etc.). The role of trust/confidence in the reported numbers is demonstrated - too little confidence can lead to a panic, societal intermingling can stop leading to other adverse effects like economic collpase and too much confidence can lead to scenarios where the confidence is unfounded, leading to unmitigated spread of the dissease.
\section{CONCLUSION }
A unified game theoretic framework incorporating the interventions of a planner, behavioral choices of individuals, detection rates and trust in reported numbers has been developed. Both the planner and population begin to favor moderate social distancing when the infection numbers begin to peak. The detection and the subsequent trust in these reported numbers also play crucial role in the spread of disease. Too little detection coupled with unfounded confidence can lead to an unmitigated spread of the disease while too little confidence when the detection is reasonably high leads to unnecessary loss of economic activity. Simulation results supporting these conclusions are presented.
\section{ACKNOWLEDGMENTS}

The author acknowledges the useful discussions he had with Prof. Sandeep K. Juneja on this topic.
This work was  supported by the Department of Atomic Energy, Government  of India, under project no. RTI4001.


\clearpage

\noindent
\textbf{Appendix A}
\\
\textit{Proof of} (\ref{adjinterp}).\\ 
In what follows we have set $P_t:=\mathbb{P}(\tau>t)$. Define $Z_t$, $\forall t \in [0,T]$ as follows :
\begin{equation*}
\begin{aligned}
Z_t&:=\mathbb{E}\bigg [\int \displaylimits_{t}^{\tau \wedge T}B g_s ds-(C+R) \one_{\{\tau \leq T\}}\bigg| \tau >t \bigg ]\\
&=\int \displaylimits_{0}^{\tau \wedge T-t} g_{r+t} \frac{P_{t+r}}{P_t}[B-\beta (C+R)g_{r+t}I_{r+t}]dr
\end{aligned}
\end{equation*}
 Observe that $Z_0$ is the loss functional $J_{\alpha}$. For any time $t<T$ we have:
\begin{equation*}
\begin{aligned}
Z_0=&\mathbb{E}\bigg [\int \displaylimits_{0}^{\tau \wedge T}B g_s ds-(C+R) \one_{\{\tau \leq T\}}\bigg]\\
=&\mathbb{E}\bigg[(\int \displaylimits_{0}^{\tau \wedge T} B g_sds -(C+R) \one_{\{\tau \leq T\}})\one_{\{\tau \leq t\}}\bigg]+\\
& \mathbb{E}\bigg[\int \displaylimits_{0}^{\tau \wedge T}B f_s ds-(C+R) \one_{\{\tau \leq T\}} \bigg| \tau >t\bigg]P_t\\
=&H_t+P_t(\int_{0}^{t}B g_s ds)+P_tZ_t
\end{aligned}
\end{equation*}
where $H_t$ is defined as 
\begin{equation*}
H_t:=\mathbb{E}[(\int_{0}^{\tau \wedge T} B g_sds -(C+R) \one_{\{\tau \leq T\}})\one_{\{\tau \leq t\}}]
\end{equation*}
Now since time $t$ was arbitrary we can write the same expression for $t+\Delta t<T$  with $\Delta t>0$. Thus:
\begin{equation*}
\begin{aligned}
H_t+P_t (\int \displaylimits_{0}^{t}B f_s ds)+P_t Z_t=&H_{t+\Delta t}+P_{t+\Delta t}(\int \displaylimits_{0}^{t+\Delta t}B g_s ds)\\
&+P_{t+\Delta t}Z_{t+\Delta t}
\end{aligned}
\end{equation*}
Solving for $Z_t$ we get:
\begin{equation*}
\begin{aligned}
Z_t&=\frac{(H_{t+\Delta t}-H_t)}{P_t}+\bigg [\frac{P_{t+\Delta t}}{P_t}(\int \displaylimits_{0}^{t+\Delta t}B g_s ds)-\int_{0}^{t}B g_s ds\bigg]\\
&+\frac{P_{t+\Delta t}}{P_t}Z_{t+\Delta t}
\end{aligned}
\end{equation*}
Plugging in the definition of $H_t$ and doing some straightforward but tedious algebra we get:
\begin{equation*}
\begin{aligned}
Z_t=&\frac{\mathbb{E}[\one_{\{t<\tau \leq t+\Delta t\}} \int_{t}^{\tau} B g_s ds]}{P_t}+\frac{P_{t+\Delta t}}{P_t}(\int_{t}^{t+\Delta t}  B g_s ds) \\
&+\frac{P_{t+\Delta t}}{P_t} Z_{t+\Delta t} -(C+R) \frac{\mathbb{E}_{\tau}[\one_{\{t<\tau \leq t+\Delta t\}} ]}{P_t}
\end{aligned}
\end{equation*}
Zubtracting $	Z_{t+\Delta t}$ on both sides and dividing by $\Delta t$ we get:
\begin{equation*}
\begin{aligned}
&\underbrace{\frac{Z_t-Z_{t+\Delta t}}{\Delta t}}_\text{A}=\underbrace{\frac{\mathbb{E}[\one_{\{t<\tau \leq t+\Delta t\}} \int_{t}^{\tau} B g_s ds]}{P_t\Delta t}}_\text{B} +\underbrace{\bigg[\frac{\frac{P_{t+\Delta t}}{P_t}-1}{\Delta t}\bigg] Z_{t+\Delta t}}_\text{C}\\
&+\underbrace{\frac{P_{t+\Delta t}}{P_t} Z_{t+\Delta t}\bigg[\frac{\int_{t}^{t+\Delta t} B g_s ds}{\Delta t}\bigg]}_\text{D} -\underbrace{(C+R) \frac{\mathbb{E}[\one_{\{t<\tau \leq t+\Delta t\}}]}{P_t\Delta t}}_\text{E}
\end{aligned}
\end{equation*}
We will analyse this in the limit $\Delta t \to 0$ term by term. For term $B$ we observe that:
\begin{equation*}
\begin{aligned}
\mathbb{E}[\one_{\{t<\tau \leq t+\Delta t\}} \int_{t}^{\tau} B g_s ds] &\leq B \mathbb{E}[\one_{\{t<\tau \leq t+\Delta t\}} (\tau-t)]\\
& \leq B \mathbb{P}(t<\tau \leq t+\Delta t) \Delta t
\end{aligned}
\end{equation*}
Hence:
\begin{equation*}
\underset{\Delta t \to 0}{lim} B=0
\end{equation*}
For term C we note that $\frac{1}{P_t}\frac{dP_t}{dt}=\underset{\Delta t \to 0}{lim}\frac{\mathbb{P}(\tau \leq t+\Delta t|\tau> t)}{\Delta t}$ to get:
\begin{equation*}
\underset{\Delta t \to 0}{lim} C=-(\beta g_t^2 I_t)Z_t
\end{equation*}
For term D from Mean Value Theorem for integrals we have for  some $t \leq s' <t+\Delta t$:
\begin{equation*}
\frac{\int_{t}^{t+\Delta t} B g_s ds}{\Delta t}=B g_{s'}
\end{equation*}
This combined with continuity of $g_t$ and the fact $\mathbb{P}(\tau>t+\Delta t|\tau> t) \to 1$ as $\Delta t \to 0$ gives us :
\begin{equation*}
\underset{\Delta t \to 0}{lim} D=B g_t
\end{equation*}
Term E can be evaluated in a manner similar to term C:
\begin{equation*}
\underset{\Delta t \to 0}{lim} E=(C+R) \beta g^2_t I_t
\end{equation*}
This shows that the limit in term A exists and is equal to $\frac{dZ_t}{dt}$. Hence we have the differential equation:
\begin{equation*}
-\frac{dZ_t}{dt}=B g_t-\beta g^2_t I_tZ_t-(C+R) \beta g^2_t I_t
\end{equation*}
We also observe from definition of $Z_t$ that $Z_T=0$. This is the same equation as the adjoint variable $\lambda_t$ in Theorem 1 with the relation $\lambda_t=-Z_t-R$. 
\end{document}